\documentclass[preprint,12pt]{elsarticle}

\usepackage{amssymb}
 \usepackage{amsthm}

\usepackage{tikz}
\usepackage{url}
\newtheorem{theorem}{Theorem}[section]
\newtheorem{lemma}[theorem]{Lemma}

\newtheorem{prop}[theorem]{Proposition}
\theoremstyle{definition}

\newtheorem{example}[theorem]{Example}

\theoremstyle{remark}

\newcommand{\F}{\mathbb F_q^*}
\newcommand{\Fq}{\mathbb F_q}
\newcommand{\Z}{\mathbb Z}

\begin{document}

\begin{frontmatter}


%% \author{Name\corref{c\fnref{label2}}

\title{Toric Codes, Multiplicative Structure and Decoding}

\author{Johan P. Hansen\fnref{label1}}
\ead{matjph@math.au.dk}
 \fntext[label2]{This work was supported by the Danish Council for Independent
Research, grant no. DFF-4002-00367}
%\address{Ny Munkegade\fnref{Ny Munkegade}}

\address{Department of Mathematics, Aarhus University}

\begin{abstract}
Long linear codes constructed from  toric varieties over  finite fields, their multiplicative structure and decoding.

The main theme is the inherent multiplicative structure on toric codes. The multiplicative structure allows for \emph{decoding},  resembling the decoding of Reed-Solomon codes and aligns with decoding by error correcting pairs.

We have used the multiplicative structure on toric codes to construct linear secret sharing schemes with \emph{strong multiplication} via Massey's construction generalizing the Shamir Linear secret sharing shemes constructed from Reed-Solomon codes. We have constructed quantum error correcting codes from toric surfaces by the Calderbank-Shor-Steane method.

\end{abstract}


\end{frontmatter}


\section{Toric varieties and codes}
\label{}

In \cite{6ffeb030f4f511dd8f9a000ea68e967b}, \cite{39bd8e90f4f211dd8f9a000ea68e967b} and \cite{53ee7c6020b511dcbee902004c4f4f50} we introduced linear codes from toric varieties and estimated the minimum distance of such codes using intersection theory. Our method to estimate the minimum distance of toric codes has subsequently been supplemented, e.g., \cite{MR2272243}, \cite{MR2476837}, \cite{MR2322944}, \cite{MR2360532}, \cite{Beelen},\cite{MR3093852} \cite{MR3345095},  and \cite{DBLP:journals/corr/Little15}.

Toric codes have an inherent multiplicative structure. 

We utilize the multiplicative structure to \emph{decode} toric codes , resembling the decoding of Reed-Solomon codes and decoding by error correcting pairs, see R. Pellikaan \cite{MR1181934} , R. K{\"o}tter \cite{RK} and  I. M\'arquez-Corbella and R. Pellikaan  Ruud\cite{MR3502016}.

The multiplicative structure on toric codes gives rise to linear secret sharing schemes with the strong multiplication property. We presented this in \cite{Hansen}  using the construction of J. L. Massey in \cite{MR2017562} and  \cite[Section 4.1]{MR2449216}. 

In \cite{MR3015727} we used toric codes  to construct quantum error correcting codes by the Calderbank-Shor-Steane method, see \cite{Calderbank19961098} and \cite{Steane19992492}.

\subsection{The construction of toric codes}
Let $\square \subset \mathbb R^r$ be an integral convex polytope.
Let $M \simeq \Z^r$ be the free $\Z$-module of rank $r$ over the integers $\Z$. For $U=\square \cap M \subseteq M$, let
$\Fq<U>$ be the linear span in $\Fq[X_1^{\pm 1},\dots, X_r^{\pm 1}]$ of the monomials
\begin{equation}
\{X^u=X_1^{u_1}\cdot \dots \cdot X_r^{u_r} \vert \ u=(u_1,\dots,u_r)\in U\}\ .
\end{equation}
This is a $\Fq$-vector space of dimension equal to the number of elements in $U$.

Let $T(\Fq)=(\F)^r$ be the $\Fq$-rational points on the torus and let $S \subseteq T(\Fq)$ be any subset. The linear map that evaluates elements in $\Fq<U>$ at all the points in $S$ is denoted by $\pi_S$:
\begin{eqnarray*}
\pi_S:\Fq<U>&\rightarrow & \Fq^{\vert S\vert}\\
 f&\mapsto&(f(P))_{P\in S}\ .
\end{eqnarray*}
In this notation $\pi_{\{P\}}(f)=f(P)$. 

Evaluting at all points in the torus  $T(\Fq)$, the 
\emph{toric code} is obtained as  the image $C=\pi_{T(\Fq)}(\Fq<U>) \subseteq  \Fq^{\vert T(\Fq)\vert} $.

\subsection{Multiplicative structure}
Toric codes inherit a certain multiplicative structure, which we used in \cite{Hansen} to obtain LSSS with strong multiplication.

Let $\square$ and $\tilde{\square}$ be polyhedra in $\mathbb R^r$, let $\square +\tilde{\square}$ denote their Minkowski sum. Let
$U=\square \cap \mathbb Z^r$ and $\tilde{U}=\tilde{\square} \cap \mathbb Z^r$. 
The map
\begin{eqnarray*}
\Fq<U> \oplus \Fq<\tilde{U}> &\rightarrow & \Fq<U+\tilde{U}>\\
 (f,g)&\mapsto&f \cdot g\ .
\end{eqnarray*}
induces a multiplication on the associated toric codes
\begin{eqnarray*}
C_{\square} \oplus C_{\tilde{\square}}  &\rightarrow & C_{\square+\tilde{\square}}\\
 (c,\tilde{c})&\mapsto&c \star \tilde{c}
\end{eqnarray*}
with coordinatewise multiplication of the codewords - the \emph{Schur} product.

\section{Multiplicative structure and decoding}

Our goal is to use the multiplicative structure to correct $t$ errors on the toric code $C_{\square}$.

This is achieved choosing another toric code $C_{\tilde{\square}}$ that helps to reduce error-correcting to a \emph{linear} problem.

Let $\square$ and $\tilde{\square}$ be polyhedra as above in $\mathbb R^2$, let $\square +\tilde{\square}$ denote their Minkowski sum.
Assume from now on:
\begin{itemize}
\item[\emph i)] $\vert \tilde{U}\vert >t$, where $\tilde{U}=\tilde{\square} \cap \mathbb Z^2$ 
\item[\emph{ii})] $d(C_{\square+\tilde{\square}}) > t$, where $d(C_{\square+\tilde{\square}})$ is the minimum distance of $C_{\square+\tilde{\square}}$.
\item[\emph{iii})]$d(C_{\tilde{\square}}) > n-d(C_{\square})$, where $d(C_{\square})$ and $d(C_{\tilde{\square}})$ are the minimum distances of $C_{\square}$ and $C_{\tilde{\square}}$.
\end{itemize}

\subsection{Error-locating}
Let the received word be $y(P)=f(P)+e(P)$ for $P\in T(\Fq)$, with $f \in \Fq<U>$ and error $e$ of Hamming-weight at most $t$ with support $T\subseteq T(\Fq)$, such that $\vert T \vert \leq t$.

From \emph{i)}, it follows that there is a $g \in \Fq<\tilde{U}>$, such that $g_{\vert T}=0$ - an \emph{error-locator}.
To find $g$, consider the linear map:
\begin{eqnarray}\label{map}
\Fq<\tilde{U}>\ \oplus\  \Fq<U+\tilde{U}>  &\rightarrow & \Fq^n\label{map}\\
 (g,h)&\mapsto&\left( g(P)y(P)-h(P)\right) _{P\in T(\Fq)}
\end{eqnarray}

As $y(P)-f(P)=0$ for $P \notin T$ (recall that the support of the error $e$ is $T$), we have that
$g(P)y(P)-(g\cdot f )(P)=0$ for all $P \in T(\Fq)$. That is $(g,h=g\cdot f)$ is in the kernel of (\ref{map}).

\begin{lemma}
Let $(g,h)$ be in the kernel of (\ref{map}). Then $g\vert T =0$ and $h=g\cdot f$.
\end{lemma}
\begin{proof} 
\begin{equation}
e(P)=y(P)-f(P)\quad \mathrm{for}\quad P\in T(\Fq)
\end{equation}
Coordinate wise multiplication yields by (\ref{map})
\begin{eqnarray*}
g(P)e(P)=&g(P)y(P)-g(P)f(P)\\=&h(P)-g(P)f(P)
\end{eqnarray*}
for $P\in T(\Fq)$.
The left hand side has Hamming weight at most $t$, the right hand side is a code word in $C_{\square+\tilde{\square}}$ with minimal distance strictly larger than $t$ by assumption \emph{ii)}. Therefore both sides equal 0.
\end{proof}
\subsection{Error-correcting}
\begin{lemma}
Let $(g,h)$ be in the kernel of (\ref{map}) with $g\vert T =0$ and $g \neq 0$. There is a unique $f$ such that $h=g\cdot f$.
\end{lemma}
\begin{proof}
As in the above proof, we have
\begin{equation}
g(P)y(P)-g(P)f(P)=0 \quad \mathrm{for}\quad P\in T(\Fq)
\end{equation}
Let $Z(g)$ be the zero-set of $g$ with $T\subseteq Z(g)$.
For $P \notin Z(g)$, we have $y(P)=f(P)$ and there are at least $d(C_{\tilde{\square}}) > n-d(C_{\square})$ such points by \emph{iii)}. This determines $f$ uniguely as it is determined by the values in $n-d(C_{\square})$ points.
\end{proof}

\begin{example}
Let $\square$ be the convex polytope with vertices $(0,0), (a,0)$ and $(0,a)$. Let  $\tilde{\square}$ be the convex polytope with vertices $(0,0), (b,0)$ and $(0,b)$. Their Minkowski sum $\square +\tilde{\square}$ is the convex polytope with vertices  $(0,0), (a+b,0)$ and $(0,a+b)$, see figure \ref{polytope2}.

From \cite[Theorem 1.3]{53ee7c6020b511dcbee902004c4f4f50}, we have that
$n=(q-1)^2,  \vert \tilde{\square}\vert = \frac{(b+1)(b+2)}{2}, d(C_{\square})=(q-1)(q-1-a), d(C_{\tilde{\square}})=(q-1)(q-1-b)$ and $ d(C_{\square+\tilde{\square}})=(q-1)(q-1-(a+b))
$ for the associated codes over $\mathbb F_q$.
 
Let $q=16, a=4$ and $b=8$. Then $n=225, \vert \tilde{\square}\vert =45, d(C_{\square})=165, d(C_{\tilde{\square}})=105$ and $ d(C_{\square+\tilde{\square}})=45
$.

As $d(C_{\tilde{\square}})=105>60=n-d(C_{\square})$, the procedure corrects $t$ errors with $t < \mathrm{Min}\left\lbrace d(C_{\square+\tilde{\square}}) , \vert \tilde{\square}\vert\right\rbrace=45$.

\begin{figure}
\caption{The convex polytope ${\square}$ with vertices $(0,0), (a,0)$ and $(0,a)$. The convex polytope ${\tilde{\square}}$ with vertices $(0,0), (b,0)$ and $(0,b))$. Their Minkowski sum $\square+\tilde{\square}$ having vertices  $(0,0), (a+b,0)$ and $(0,a+b)$.}\label{polytope2}
\begin{center}
\begin{tikzpicture}[scale=0.35]

\draw[very thick] (0,0)--(14,0);
\draw[very thick] (0,0)--(0,14);
\draw[thick, fill=black!5] (0,0)--(10,0)--(0,10)--(0,0);
\draw[thick,fill=black!10] (0,0)--(12,0)--(0,12)--(0,0);
\draw (12,0) node[anchor=north]{$a+b$};
\draw (0,12) node[anchor=east]{$a+b$};
\draw[thick,fill=black!10] (0,0)--(8,0)--(0,8)--(0,0);
\draw[thick,fill=black!15] (0,0)--(4,0)--(0,4)--(0,0);
\draw (8,0) node[anchor=north]{$b$};
\draw (0,8) node[anchor=east]{$b$};
\draw (4,0) node[anchor=north]{$a$};
\draw (0,4) node[anchor=east]{$a$};
\draw (4,0) node[anchor=north]{$a$};
\draw (0,4) node[anchor=east]{$a$};
\draw (14,0) node[anchor=west]{$q-2$};
\draw (0,14) node[anchor=east]{$q-2$};
\draw[step=1cm,gray,very thin,dashed]
(0,0) grid (14,14);
\end{tikzpicture}
\end{center}
\end{figure}
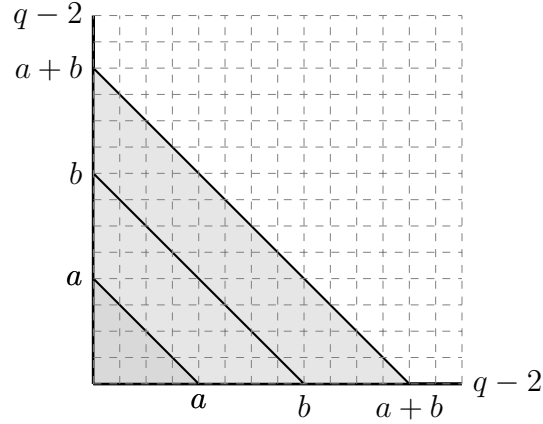
\end{example}
\subsection{Error correcting pairs}
R. Pellikaan \cite{MR1181934} and  R. K{\"o}tter \cite{RK} introduced the concept of error correcting pairs for a linear code, see also I. M\'arquez-Corbella and R. Pellikaan  Ruud\cite{MR3502016}. Specifically for a linear code $C \subseteq  \mathbb F_q^n$  an $t$-error correcting pair consists of two linear codes $A,B \subseteq  \mathbb F_q^n$, such that
\begin{equation}
(A\star B) \perp C, \dim_{\mathbb F_q} A > t, d(B^\perp)>t, d(A)+d(C)>n
\end{equation}
Here $A\star B = \{a \star b \vert\ a \in A, b \in B\}$ and 
$\perp$ denotes ortogonality with respect to the usual inner product. They described the known decoding algorithms  for decoding $t$ or fewer errors in this framwork.

Also the decoding in the present paper can be described in this framework, taking $C=C_{\square}, A=C_{\tilde{\square}}$ and $B=(C\star A)^{\perp}$ using Proposition \ref{orto}.

\subsubsection{Orthogonality - dual code}
In Proposition \ref{orto} we present the dual code of $C=\pi_S(\Fq<U>)$.

Let $U \subseteq M$ be a subset, define its opposite as $-U:= \{-u \vert \ u \in U \} \subseteq M$.
The opposite maps the monomial $X^u$ to $X^{-u}$ and induces by linearity an isomorphism of vector spaces 
\begin{eqnarray*}
\Fq<U>&\rightarrow& \Fq<-U>\\
X^u&\mapsto & X^{-u}\\
f&\mapsto&\hat{f}\ .
\end{eqnarray*}
On $\Fq^{\vert T(\Fq) \vert}$, we have the usual inner product 
\begin{equation}
(a_0,\dots,a_n)\cdot(b_0,\dots,b_n)=\sum_{l=0}^n a_l b_l \in \Fq\ ,
\end{equation}
with $n=\vert T(\Fq) \vert-1$.
\begin{lemma} Let $f,g \in \Fq <M> $ and assume $f \neq \hat{g}$, then
\begin{equation}
\pi_{T(\Fq)}(f)\cdot \pi_{T(\Fq)}(g)= 0 
\end{equation}
\end{lemma}

Let 
\begin{equation}
H =\{0,1,\dots,q-2\}\times \dots \times \{0,1,\dots,q-2\} \subset M\ .
\end{equation}

With this inner product we obtain the following proposition, e.g. \cite[Proposition 3.5]{MR2377234} and \cite[Theorem 6]{MR2499927}.
\begin{prop}\label{orto}
Let $U \subseteq H$ be a subset. Then we have
\begin{itemize}
\item[{\it i)}] For $f \in \Fq<U>$ and $g \notin \Fq<-H\setminus -U>$, we have that
$\pi_{T(\Fq)}(f)\cdot \pi_{T(\Fq)}(g)=0$. 

\item[{\it ii)}] The orthogonal complement to $\pi_{T(\Fq)}(\Fq<U>)$ in $\Fq^{\vert {T(\Fq)} \vert}$ is 
\begin{equation}
\pi_{T(\Fq)}(\Fq<-H \setminus -U>)\ ,
\end{equation}
i.e., the dual code of $C=\pi_{T(\Fq)}(\Fq<U>)$ is $\pi_{T(\Fq)}(\Fq<-H \setminus -U>)$.
\end{itemize}
\end{prop}



\section*{References}
\bibliographystyle{elsarticle-num} 
\bibliography{SSS}

\begin{thebibliography}{10}
\expandafter\ifx\csname url\endcsname\relax
  \def\url#1{\texttt{#1}}\fi
\expandafter\ifx\csname urlprefix\endcsname\relax\def\urlprefix{URL }\fi
\expandafter\ifx\csname href\endcsname\relax
  \def\href#1#2{#2} \def\path#1{#1}\fi

\bibitem{6ffeb030f4f511dd8f9a000ea68e967b}
J.~Hansen, Toric surfaces and codes, in: Information Theory Workshop, IEEE,
  1998, pp. 42--43.
\newblock \href {http://dx.doi.org/10.1109/ITW.1998.706405}
  {\path{doi:10.1109/ITW.1998.706405}}.

\bibitem{39bd8e90f4f211dd8f9a000ea68e967b}
J.~Hansen, Toric surfaces and error-correcting codes, in: J.~Buchmann,
  T.~Hoeholdt, H.~Stichtenoth, H.~Tapia-Recillas (Eds.), Coding theory,
  cryptography and related areas, Springer, 2000, pp. 132--142.

\bibitem{53ee7c6020b511dcbee902004c4f4f50}
J.~Hansen, Toric varieties {H}irzebruch surfaces and error-correcting codes,
  Applicable Algebra in Engineering, Communication and Computing 13~(4) (2002)
  289--300.

\bibitem{MR2272243}
J.~Little, H.~Schenck, \href{http://dx.doi.org/10.1137/050637054}{Toric surface
  codes and {M}inkowski sums}, SIAM J. Discrete Math. 20~(4) (2006) 999--1014
  (electronic).
\newblock \href {http://dx.doi.org/10.1137/050637054}
  {\path{doi:10.1137/050637054}}.
\newline\urlprefix\url{http://dx.doi.org/10.1137/050637054}

\bibitem{MR2476837}
I.~Soprunov, J.~Soprunova, \href{http://dx.doi.org/10.1137/080716554}{Toric
  surface codes and {M}inkowski length of polygons}, SIAM J. Discrete Math.
  23~(1) (2008/09) 384--400.
\newblock \href {http://dx.doi.org/10.1137/080716554}
  {\path{doi:10.1137/080716554}}.
\newline\urlprefix\url{http://dx.doi.org/10.1137/080716554}

\bibitem{MR2322944}
J.~Little, R.~Schwarz, On toric codes and multivariate {V}andermonde matrices,
  Appl. Algebra Engrg. Comm. Comput. 18~(4) (2007) 349--367.
\newblock \href {http://dx.doi.org/10.1007/s00200-007-0041-1}
  {\path{doi:10.1007/s00200-007-0041-1}}.

\bibitem{MR2360532}
D.~Ruano, \href{http://dx.doi.org/10.1016/j.ffa.2007.02.002}{On the parameters
  of {$r$}-dimensional toric codes}, Finite Fields Appl. 13~(4) (2007)
  962--976.
\newblock \href {http://dx.doi.org/10.1016/j.ffa.2007.02.002}
  {\path{doi:10.1016/j.ffa.2007.02.002}}.
\newline\urlprefix\url{http://dx.doi.org/10.1016/j.ffa.2007.02.002}

\bibitem{Beelen}
P.~Beelen, D.~Ruano, \href{http://dx.doi.org/10.1007/978-3-642-02181-7_1}{The
  order bound for toric codes}, in: M.~Bras-Amorós, T.~H\o{}holdt (Eds.),
  Applied Algebra, Algebraic Algorithms and Error-Correcting Codes, Vol. 5527
  of Lecture Notes in Computer Science, Springer Berlin Heidelberg, 2009, pp.
  1--10.
\newline\urlprefix\url{http://dx.doi.org/10.1007/978-3-642-02181-7_1}

\bibitem{MR3093852}
J.~B. Little, \href{http://dx.doi.org/10.1016/j.ffa.2013.05.004}{Remarks on
  generalized toric codes}, Finite Fields Appl. 24 (2013) 1--14.
\newblock \href {http://dx.doi.org/10.1016/j.ffa.2013.05.004}
  {\path{doi:10.1016/j.ffa.2013.05.004}}.
\newline\urlprefix\url{http://dx.doi.org/10.1016/j.ffa.2013.05.004}

\bibitem{MR3345095}
I.~Soprunov, \href{http://dx.doi.org/10.13069/jacodesmath.75353}{Lattice
  polytopes in coding theory}, J. Algebra Comb. Discrete Struct. Appl. 2~(2)
  (2015) 85--94.
\newblock \href {http://dx.doi.org/10.13069/jacodesmath.75353}
  {\path{doi:10.13069/jacodesmath.75353}}.
\newline\urlprefix\url{http://dx.doi.org/10.13069/jacodesmath.75353}

\bibitem{DBLP:journals/corr/Little15}
J.~B. Little, \href{http://arxiv.org/abs/1504.07494}{Toric codes and finite
  geometries}, arxiv abs/1504.07494.
\newline\urlprefix\url{http://arxiv.org/abs/1504.07494}

\bibitem{MR1181934}
R.~Pellikaan, On decoding by error location and dependent sets of error
  positions, Discrete Math. 106/107 (1992) 369--381, a collection of
  contributions in honour of Jack van Lint.

\bibitem{RK}
R.~K{\"o}tter, A unified description of an error locating procedure for linear
  codes, in: Proceedings of Algebraic and Combinatorial Coding Theory, Voneshta
  Voda, 1992, pp. 113--117.

\bibitem{MR3502016}
I.~M\'arquez-Corbella, R.~Pellikaan,
  \href{http://dx.doi.org/10.1016/j.ffa.2016.04.004}{A characterization of
  {MDS} codes that have an error correcting pair}, Finite Fields Appl. 40
  (2016) 224--245.
\newblock \href {http://dx.doi.org/10.1016/j.ffa.2016.04.004}
  {\path{doi:10.1016/j.ffa.2016.04.004}}.
\newline\urlprefix\url{http://dx.doi.org/10.1016/j.ffa.2016.04.004}

\bibitem{Hansen}
J.~P. {Hansen}, {Secret Sharing Schemes with a large number of players from
  Toric Varieties}, ArXiv e-prints\href {http://arxiv.org/abs/1410.4378}
  {\path{arXiv:1410.4378}}.

\bibitem{MR2017562}
J.~L. Massey, Some applications of code duality in cryptography, Mat. Contemp.
  21 (2001) 187--209, 16th School of Algebra, Part II (Portuguese)
  (Bras{\'{\i}}lia, 2000).

\bibitem{MR2449216}
H.~Chen, R.~Cramer, S.~Goldwasser, R.~de~Haan, V.~Vaikuntanathan, Secure
  computation from random error correcting codes, in: Advances in
  cryptology---{EUROCRYPT} 2007, Vol. 4515 of Lecture Notes in Comput. Sci.,
  Springer, Berlin, 2007, pp. 291--310.

\bibitem{MR3015727}
J.~P. Hansen, \href{http://dx.doi.org/10.1109/TIT.2012.2220523}{Quantum codes
  from toric surfaces}, IEEE Trans. Inform. Theory 59~(2) (2013) 1188--1192.
\newblock \href {http://dx.doi.org/10.1109/TIT.2012.2220523}
  {\path{doi:10.1109/TIT.2012.2220523}}.
\newline\urlprefix\url{http://dx.doi.org/10.1109/TIT.2012.2220523}

\bibitem{Calderbank19961098}
A.~Calderbank, P.~Shor, Good quantum error-correcting codes exist, Physical
  Review A - Atomic, Molecular, and Optical Physics 54~(2) (1996) 1098--1105.

\bibitem{Steane19992492}
A.~Steane, Enlargement of calderbank-shor-steane quantum codes, IEEE
  Transactions on Information Theory 45~(7) (1999) 2492--2495.

\bibitem{MR2377234}
M.~Bras-Amor{\'o}s, M.~E. O'Sullivan,
  \href{http://dx.doi.org/10.3934/amc.2008.2.15}{Duality for some families of
  correction capability optimized evaluation codes}, Adv. Math. Commun. 2~(1)
  (2008) 15--33.
\newblock \href {http://dx.doi.org/10.3934/amc.2008.2.15}
  {\path{doi:10.3934/amc.2008.2.15}}.
\newline\urlprefix\url{http://dx.doi.org/10.3934/amc.2008.2.15}

\bibitem{MR2499927}
D.~Ruano, \href{http://dx.doi.org/10.1016/j.jsc.2007.07.018}{On the structure
  of generalized toric codes}, J. Symbolic Comput. 44~(5) (2009) 499--506.
\newblock \href {http://dx.doi.org/10.1016/j.jsc.2007.07.018}
  {\path{doi:10.1016/j.jsc.2007.07.018}}.
\newline\urlprefix\url{http://dx.doi.org/10.1016/j.jsc.2007.07.018}

\end{thebibliography}


%\begin{thebibliography}{00}


%\bibitem{}
%\end{thebibliography}
\end{document}